\documentclass[11pt]{article}%
\usepackage{amsfonts}
\usepackage{amsmath}
\usepackage{amssymb}
\usepackage{graphicx}%
\providecommand{\U}[1]{\protect\rule{.1in}{.1in}}
\newtheorem{theorem}{Theorem}
\newtheorem{acknowledgement}[theorem]{Acknowledgement}

\newtheorem{example}[theorem]{Example}

\newtheorem{lemma}[theorem]{Lemma}

\newtheorem{proposition}[theorem]{Proposition}
\newtheorem{remark}[theorem]{Remark}

\newenvironment{proof}[1][Proof]{\noindent\textbf{#1.} }{\ \rule{0.5em}{0.5em}}
\begin{document}

\title{On the Energy Minimizing Problem in Plasma Dynamics}
\author{Maurice de Gosson\\University Of Vienna\\Faculty of Mathematics (NuHAG)\\Oskar-Morgenstern-Platz 1, 1090 Vienna}
\maketitle

\begin{abstract}
A fundamental question in plasma physics and its applications to fusion theory
is to find minimum energy stares of plasma.. Hong Qin has suggested the use of
what he calls "Gromov ground states" which take into account the phenomenon of
symplectic rigidity, as opposed to Gardner's classical approach which only use
volume-preserving diffeomorphisms. Howev4er Qin's approach leads to
topological difficulties because his constructions remain in the classical
(Hamiltonian) domain. We circumvent these difficulties by introducing a
quantum-mechanical formulation of the problem, based on the Wigner transform
and the symplectic geometry of Gaussian state

\end{abstract}

Dedicated to the memory of Basil Hiley (1935--2025)

\section{Introduction}

In their study \cite{Qin1,Qin2} of plasma dynamics Hong Qin and his
collaborators ask the following question: let $\Omega$ be a compact, connected
subset of the standard symplectic phase space $(\mathbb{R}^{2n},\sigma)$, and
let $H:\mathbb{R}^{2n}\rightarrow\mathbb{R}$ be a smooth Hamiltonian function.
Under what conditions (if any) does there exist a symplectomorphism
$f\in\operatorname{Symp}(n)$ such that
\begin{equation}
\int_{f(\Omega)}H(z)\,\mathrm{d}z<\int_{\Omega}H(z)\,\mathrm{d}z
\label{original}%
\end{equation}
and can we determine the lowest energy attainable from $\Omega$ by such an
$f$, \textit{i.e.} what Qin \textit{et al.} call the \textquotedblleft Gromov
energy\textquotedblright\
\begin{equation}
E_{\mathrm{Gromov}}(\Omega;H)=\inf_{f\in\operatorname{Symp}(n)}\int%
_{f(\Omega)}H(z)\,\mathrm{d}z\,? \label{gromovdomain}%
\end{equation}
Their question was motivated by the observation that Hamiltonian plasma
dynamics satisfies considerably stronger geometric constraints than mere
preservation of phase-space volume as considered by Gardner in his historical
paper \cite{Gardner}.

It turns out, as we will see, that the answer to Qin's question is in general
negative. This \textquotedblleft no-go\textquotedblright\ result leads us to
replace Hong Qin's initial question with the following one: Given a
Hamiltonian function $H\in C^{\infty}(\mathbb{R}^{2n}\times\mathbb{R}%
_{t},\mathbb{R})$ and a Wigner distribution $\rho\in L^{2}(\mathbb{R}^{2n})$,
does there exist $f\in\operatorname{Symp}(n)$ such that
\begin{equation}
\int_{\mathbb{R}^{2n}}H(z,t)\rho(z)\,\mathrm{d}z>\int_{\mathbb{R}^{2n}%
}H(z,t)\rho(f(z))\,\mathrm{d}z\,? \label{wignerquestion}%
\end{equation}
We will see that for quadratic Hamiltonians this question is closely related
to Qin's initial question; the main difference is that our reformulation
places the problem in the quantum setting, as opposed to Qin's framework which
is classical Hamiltonian mechanics.

\paragraph{Notation.}

The standard symplectic form
\[
\sigma=\sum_{j=1}^{n}\mathrm{d}p_{j}\wedge\mathrm{d}x_{j}%
\]
is written in matrix form as $\sigma(z,z^{\prime})=Jz\cdot z^{\prime}$, where
\[
J=%
\begin{pmatrix}
0_{n\times n} & I_{n\times n}\\
-I_{n\times n} & 0_{n\times n}%
\end{pmatrix}
\]
is the standard symplectic matrix. The symplectic group $\operatorname{Sp}(n)$
is the group of automorphisms $S$ of $\mathbb{R}^{n}\times\mathbb{R}^{n}%
\equiv\mathbb{R}^{2n}$ such that $S^{\ast}\sigma=\sigma$; in matrix notation:
$S\in\operatorname{Sp}(n)$ if and only if $S^{T}JS=J$ (or $SJS^{T}=J$),
$S^{T}$ being the transpose of $S$. We denote by $\operatorname{Symp}(n)$ the
group of all symplectomorphisms and by $\operatorname{Ham}(n)$ its subgroup
consisting of Hamiltonian diffeomorphisms.

\section{\label{sec1}An Example, and a Topological Obstruction}

\subsection{Symplectic non-squeezing}

We recall Gromov's symplectic non-squeezing theorem \cite{gr85}-Let
$B^{2n}(r)$ be the centered phase-space ball
\[
B^{2n}(r)=\{z\in\mathbb{R}^{2n}:|z|\leq r\}
\]
and let
\[
Z_{j}^{2n}(R)=\{(x,p)\in\mathbb{R}^{2n}:x_{j}^{2}+p_{j}^{2}\leq R^{2}\}
\]
be%
\[
B^{2n}(r)\hookrightarrow Z_{j}^{2n}(R)
\]
exists if and only if $r\leq R$. Gromov's theorem implies what is sometimes
called the "Principle of the Symplectic Camel": Let $P_{j}$ be the plane of
conjugate coordinates $x_{j},p_{j}$ in $\mathbb{R}^{2n}$ and $D_{j}$ be a
circular disk with radius $r$ in $P_{j}$. There exists a smooth one-parameter
family $(f_{t})_{t\in\lbrack0,1]}$ of symplectomorphisms of $(\mathbb{R}%
^{2n},\sigma)$ taking an arbitrary phase-space ball $B^{2n}(z,R)$ with radius
$R$ through the disk $D_{j}$ if and only if $R\leq r$. A crucial consequence
of Gromov's non-squeezing theorem is the principle of symplectic rigidity,
which says that volume-preserving diffeomorphisms cannot in general be
approximated by symplectomorphisms (Gromov--Eliashberg theorem \cite{Elia,HZ}%
).The group $\operatorname{Symp}(n)$ is $C^{0}$-closed in $\operatorname{Diff}%
(n)$. Here is a sketch of the proof: let $\phi_{j}\in\operatorname{Symp}(n)$,
$\phi_{j}\rightarrow\phi$, and assume that $\phi$ is a $C^{1}$-diffeomorphism.
We prove that $\phi\in\operatorname{Symp}(n)$. Suppose, on the contrary, that
$\phi$ is not symplectic. Then at some point $z_{0}$ we have $D\phi
(z_{0})^{\ast}\sigma\neq\sigma$. Using the local linearization at $z_{0}$ and
composing with suitable linear symplectic transformations, one obtains a
sufficiently small ball $B^{2n}(r;z_{0})$ whose image by $\phi$ is contained
in a symplectic cylinder of smaller radius:
\[
\phi(B^{2n}(r;z_{0}))\subset Z_{1}^{2n}(R),\qquad R<r.
\]
Since $\phi_{j}\rightarrow\phi$ uniformly, after slightly decreasing $r$ and
slightly increasing $R$ if necessary, the same inclusion holds for all
sufficiently large $j$:
\[
\phi_{j}(B^{2n}(r;z_{0}))\subset Z_{1}^{2n}(R),\qquad R<r.
\]
But $\phi_{j}\in\operatorname{Symp}(n)$, so its restriction gives a symplectic
embedding $B^{2n}(r;z_{0})\hookrightarrow Z_{1}^{2n}(R)$ with $R<r$. This
contradicts Gromov's non-squeezing theorem.

\subsection{The ``water-bag'' example}

Gardner's question \cite{Gardner} may be phrased as follows: for a prescribed
particle distribution and Hamiltonian, what is the minimum energy obtainable
by rearrangements compatible with collisionless dynamics? In the
volume-preserving relaxation, particles may be moved toward low-energy regions
while preserving phase-space volume. Qin and collaborators \cite{Qin2}
illustrate the essential difficulty in the case $n=2$. Denoting by
$(x,y,p_{x},p_{y})\in\mathbb{R}^{4}$ the phase-space variable, they introduce
the Hamiltonian function
\[
H(x,y,p_{x},p_{y})=\frac{1}{2}(p_{x}^{2}+p_{y}^{2}+x^{2}),
\]
which is a simplified model for charged-particle dynamics in accelerators, and
particles initially occupying uniformly the phase-space ball $B^{4}(r)$ (it is
called a a \textquotedblleft water-bag distribution\textquotedblright\ in
\cite{Qin2}). A volume-preserving diffeomorphism can send this ball to a long
thin region of the form
\[
B^{3}(R)\times\lbrack0,L],
\]
where the $B^{3}(R)$ factor involves the variables $(x,p_{x},p_{y})$. Keeping
the four-dimensional volume fixed while decreasing $R$ allows the energy to be
made arbitrarily small by increasing the extent in the energetically neutral
$y$ direction. The situation changes completely if the diffeomorphism is
required to be symplectic. Since the target is contained in the cylinder
\[
\{(x,y,p_{x},p_{y}):x^{2}+p_{x}^{2}\leq R^{2}\},
\]
Gromov's non-squeezing theorem requires $R\geq r$. Hence the volume-preserving
Gardner relaxation may predict states that cannot be reached by Hamiltonian dynamics.

\subsection{A topological obstruction: Polterovich's argument}

The following obstruction to the existence of an energy-minimizing Hamiltonian
image is due to Leonid Polterovich (private communication).

Let $(M,\sigma)$ be a symplectic surface, $D\subset M$ a smoothly embedded
closed disc, and $\gamma=\partial D$. For a smooth function $H:M\rightarrow
\mathbb{R}$, consider
\begin{equation}
\mathcal{E}(f)=\int_{f(D)}H\,\sigma,\qquad f\in\operatorname{Ham}(M,\sigma).
\label{energy}%
\end{equation}
Suppose that $D$ is a critical point of this functional under Hamiltonian
variations (in particular, this holds if $D$ minimizes or maximizes the
integral among its Hamiltonian images). Let $F$ be an arbitrary smooth
Hamiltonian, with Hamiltonian vector field $X_{F}$ and flow $f_{t}$, and adopt
the convention $\iota_{X_{F}}\sigma=-dF$. Cartan's formula and Stokes' theorem
give
\begin{align}
\left.  \frac{d}{dt}\right\vert _{t=0}\int_{f_{t}(D)}H\,\sigma &  =\int%
_{D}\mathcal{L}_{X_{F}}(H\sigma)=\int_{D}d\left(  \iota_{X_{F}}(H\sigma
)\right) \nonumber\\
&  =\int_{\gamma}H\,\iota_{X_{F}}\sigma=-\int_{\gamma}H\,dF=\int_{\gamma
}F\,dH. \label{firstvariation}%
\end{align}
Here $d(H\sigma)=0$ automatically on a surface. At a critical point, the final
integral vanishes for every $F$. Since the restriction of $F$ to $\gamma$ can
be chosen arbitrarily, $dH$ annihilates the tangent bundle $T\gamma$. Thus $H$
is constant along the connected curve $\gamma$, say
\begin{equation}
\gamma\subset H^{-1}(h_{0}). \label{level}%
\end{equation}
Note that this does \emph{not} imply $dH=0$ as a covector on $\gamma$: only
its tangential restriction vanishes. If $h_{0}$ is a regular value, $\gamma$
is a connected component of the level set.

We now specialize to $M=S^{2}$, with total symplectic area normalized to $1$.
Choose a Morse function $H:S^{2}\rightarrow\mathbb{R}$ whose Reeb graph is a
tripod, with one trivalent vertex corresponding to a saddle and three terminal
branches. Denote by $a,b,c$ the symplectic areas of the three regions
associated with these branches, and assume
\begin{equation}
a+b+c=1,\qquad0<a\leq b\leq c<\tfrac{1}{2}. \label{areas}%
\end{equation}
A regular level circle corresponding to a point on the branch of area $a$ cuts
off, on the terminal side, a disc of area at most $a$; the opposite disc has
area at least $1-a$. The analogous statement holds on the other two branches.
Consequently, every regular level circle bounds discs of areas belonging to
\begin{equation}
\lbrack0,c]\cup\lbrack1-c,1]. \label{allowed}%
\end{equation}
The same conclusion holds for smooth embedded circles contained in critical
levels: such a circle cannot pass smoothly through a nondegenerate saddle
while remaining in the critical level, and any remaining smooth circle is on a
regular portion of that level.

\begin{proposition}
Let $D\subset S^{2}$ be a smoothly embedded disc of area $\mu$ with
\begin{equation}
c<\mu<1-c. \label{gap}%
\end{equation}
Then the functional \ref{energy} has no critical point on the Hamiltonian
orbit of $D$, and in particular has no minimizer or maximizer there.
\end{proposition}

Hamiltonian diffeomorphisms preserve symplectic area, so every $f(D)$ has area
$\mu$. If $f(D)$ were critical, the first-variation argument applied at $f(D)$
would force its smooth boundary to lie in a level set of $H$. But
(\ref{allowed}) shows that no smooth level circle bounds a disc of area $\mu$
satisfying (\ref{gap}). This is a contradiction.

The conclusion concerns \emph{nonattainment}: the infimum of (\ref{energy})
still exists when $H$ is bounded, but need not be achieved by any Hamiltonian
image of $D$. Thus the obstruction is topological, encoded by the Reeb graph
and the areas of its branches, rather than a failure of the first-variation
calculation. In particular, one cannot infer existence of an unrestricted
energy-minimizing \textquotedblleft Gromov ground state\textquotedblright%
\ from symplectic area preservation alone. Existence questions for restricted
families, such as Gaussian states or symplectic images of quantum blobs,
require separate arguments. Note the following Sign convention and extremum
direction\textbf{: }the inequality $\int_{D}H\,\sigma\geq\int_{f(D)}H\,\sigma$
for all $f$ characterizes a \emph{maximum} at $D$; for a minimum the
inequality is reversed. Either extremum must satisfy firs tvariation with zero
first derivative.

\section{Gaussian Distributions\label{sec2}}

\subsection{An averaging result}

Let $\Omega\subset\mathbb{R}^{2n}$ be a centered phase space ellipsoid:%
\[
\Omega=\left\{  z\in\mathbb{R}^{2n}:\frac{1}{2}\Sigma^{-1}z\cdot z<1\right\}
\]
with $\Sigma=\Sigma^{T}>0$. Let
\begin{equation}
\rho(z)=\frac{1}{(2\pi)^{n}\sqrt{\det\Sigma}}\exp\left(  -\frac{1}{2}%
\Sigma^{-1}z\cdot z\right)  ,\quad z\in\mathbb{R}^{2n} \label{Gauss}%
\end{equation}
be~a Gaussian probability distribution on $\mathbb{R}^{n}$ with covariance
matrix $\Sigma$ and covariance ellipsoid $\Omega$.

\begin{proposition}
\label{propmoy}Let $\Omega$ be a non-degenerate phase space ellipsoid and
$\rho_{\Omega}$ the associated Gaussian (\ref{Gauss}); Let $H\in C^{\infty
}(\mathbb{R}_{z}^{2n},\mathbb{R})$ be a Hamiltonian function. (i) If $H$ is
homogeneus of degree $k>0$ then we have%
\begin{equation}
\int_{\mathbb{R}^{2n}}H(z)\rho_{\Omega}(z)\,\mathrm{d}z=\frac{\Gamma\left(
n+\frac{k}{2}+1\right)  }{n!}\frac{1}{|\Omega|}\int_{\Omega}H(z)\,\mathrm{d}z.
\label{av1}%
\end{equation}
(ii) In particular, for $k=2$,%
\begin{equation}
\int_{\mathbb{R}^{2n}}H(z)\rho(z)\,\mathrm{d}z=\frac{n+1}{|\Omega|}%
\int_{\Omega}H(z)\,\mathrm{d}z. \label{av2}%
\end{equation}

\end{proposition}

\begin{proof}
Put $q(z)=\frac{1}{2}\Sigma^{-1}z\cdot z$. Since
\[
e^{-q(z)}=\int_{q(z)}^{\infty}e^{-t}\,\mathrm{d}t
\]
and $q(z)<t$ if and only if $z\in\sqrt{t}\,\Omega$, Fubini's theorem gives
\[
\int_{\mathbb{R}^{2n}}H(z)\rho_{\Omega}(z)\,\mathrm{d}z=\frac{1}{(2\pi
)^{n}\sqrt{\det\Sigma}}\int_{0}^{\infty}e^{-t}\left(  \int_{\sqrt{t}\,\Omega
}H(z)\,\mathrm{d}z\right)  \mathrm{d}t.
\]
In words, the Wigner expectation of $H$ is a weighted average of the integrals
of $H$ over all homothetic dilations $\sqrt{t}\,\Omega$ of the covariance
ellipsoid. If $H$ is homogeneous of degree $k$, $H(\lambda z)=\lambda^{k}%
H(z)$, then
\[
\int_{\sqrt{t}\,\Omega}H(z)\,\mathrm{d}z=t^{n+k/2}\int_{\Omega}%
H(z)\,\mathrm{d}z,
\]
and consequently
\[
\int_{\mathbb{R}^{2n}}H(z)\rho_{\Omega}(z)\,\mathrm{d}z=\frac{\Gamma\left(
n+\frac{k}{2}+1\right)  }{(2\pi)^{n}\sqrt{\det\Sigma}}\int_{\Omega
}H(z)\,\mathrm{d}z.
\]
Since
\[
|\Omega|=\frac{(2\pi)^{n}}{n!}\sqrt{\det\Sigma},
\]
this becomes
\[
\int_{\mathbb{R}^{2n}}H(z)\rho(z)\,\mathrm{d}z=\frac{\Gamma\left(  n+\frac
{k}{2}+1\right)  }{n!}\frac{1}{|\Omega|}\int_{\Omega}H(z)\,\mathrm{d}z.
\]
For $k=2$ this reduces to
\[
\int_{\mathbb{R}^{2n}}H(z)\rho(z)\,\mathrm{d}z=(n+1)\frac{1}{|\Omega|}%
\int_{\Omega}H(z)\,\mathrm{d}z.
\]

\end{proof}

Writing
\begin{equation}
\langle H\rangle_{\Omega}=\frac{1}{|\Omega|}\int_{\Omega}H(z)\,\mathrm{d}z,
\label{av3}%
\end{equation}
formula (\ref{av2}) can be written
\begin{equation}
\int_{\mathbb{R}^{2n}}H(z)\rho(z)\,\mathrm{d}z=(n+1)\langle H\rangle_{\Omega}
\label{av4}%
\end{equation}
when $H$ is a quadratic Hamiltonian.

\subsection{The quantum case}

Recall that the Wigner transform (or function) $W\psi$ of $\psi\in
L^{2}(\mathbb{R}^{n})$ is defined by
\begin{equation}
W\psi(z)=\left(  \tfrac{1}{2\pi\hbar}\right)  ^{n}\int_{\mathbb{R}^{n}%
}e^{-\frac{i}{\hbar}p\cdot y}\psi(x+\tfrac{1}{2}y)\psi^{\ast}(x-\tfrac{1}%
{2}y)dy. \label{wig}%
\end{equation}
It is often convenient to use the alternative definition \cite{WIGNER}%
\begin{equation}
W\psi(z)=\left(  \tfrac{1}{\pi\hbar}\right)  ^{n}(\widehat{R}(z)\psi
|\psi)_{L^{2}} \label{wig2}%
\end{equation}
where $\widehat{R}(z)$ is the Grossmann--Royer reflection operator
\cite{WIGNER}%
\begin{equation}
\widehat{R}(z)=\widehat{T}(z)\widehat{R}\widehat{T}(z)^{-1}\text{ \ ,
\ }\widehat{R}(x)\psi=\psi(-x); \label{GR}%
\end{equation}
where $\widehat{T}(z)=e^{-\frac{it}{\hslash}\sigma(\hat{z},z_{0})}$ is the
Heisenberg--Weyl displacement operator%
\begin{equation}
\psi(x,t)=\psi_{0}(x)\psi_{0}(x)\text{.} \label{psixt}%
\end{equation}
exPlicitly given by
\begin{equation}
\widehat{T}(z_{0})\psi(x)=e^{\tfrac{i}{\hslash}(p_{0}\cdot x-\tfrac{1}{2}%
p_{0}\cdot x_{0})}\psi_{0}(x-x_{0}). \label{HW}%
\end{equation}

We now assume that the covariance matrix $\Sigma$ in (\ref{Gauss}) satisfies
the the "quantum condition"%
\begin{equation}
\Sigma+\frac{i\hbar}{2}J\geq0 \label{quantum}%
\end{equation}
which implies \cite{Birkbis} that the Gaussian function i the Wigner
distribution of a \ quantum state Condition (\ref{quantum}) is equivalent to
the Robertson--Schr\"{o}dinger uncertainty relations, and implies tat the
covariance ellipsoid $\Omega$ contains a "quantum blob", i.e. a symplectic
ball $S(B^{2n}(\sqrt{\hbar})$, $S\in\operatorname*{Sp}(n)$.

Recall the "Wigner--Weyl" identity \cite{Birkbis,WIGNER,Littlejohn}: .

\begin{lemma}
\label{Lemma1}Let $\widehat{H}=\operatorname*{Op}_{\mathrm{W}}(H)$ be the Weyl
quantization of\ $H$; we assume that $\widehat{H}:\mathcal{S}(\mathbb{R}%
^{n})\rightarrow\hookrightarrow\mathcal{S}(\mathbb{R}^{n})$. We have, or
$\psi\in\mathcal{S}(\mathbb{R}^{n})$,
\[
(\widehat{H}\psi|\psi)_{L^{2}}=\int_{\mathbb{R}^{n}}H(z)W\psi(z)z
\]

\end{lemma}

\begin{proof}
Rewriting the textbook formula%
\[
\widehat{H}\psi(x)=\left(  \tfrac{1}{2\pi\hbar}\right)  ^{n}\int%
_{\mathbb{R}^{n}}e^{\frac{i}{\hbar}p\cdot(x-y)}H(x+\tfrac{1}{2}y,x-\tfrac
{1}{2}y)\psi(y)dy
\]
as (see \cite{WIGNER})
\[
\widehat{H}\psi=\left(  \tfrac{1}{\pi\hbar}\right)  ^{n}\int_{\mathbb{R}^{2n}%
}H(z_{0})\widehat{R}(z_{0})\psi dz_{0}.
\]
where $\widehat{R}(z)$ is the Grossmann--Royer operator (\ref{GR}), we have,
using (\ref{wig2})
\begin{align*}
(\widehat{H}\psi|\psi)_{L^{2}}  &  =\left(  \tfrac{1}{\pi\hbar}\right)
^{n}\int_{\mathbb{R}^{2n}}H(z)(\widehat{R}(z)\psi|\psi)_{L^{2}}dz\\
&  =\int_{\mathbb{R}^{n}}H(z)W\psi(z)dz.
\end{align*}

\end{proof}

We have:

\begin{proposition}
Assume that the Gaussian $\rho$ satisfies the condition (\ref{quantum}). There
exists an orthonormal sequence of functions $\psi_{j}\in L^{2}(\mathbb{R}%
^{n})$ and numbers $\lambda_{j}\geq0$, $\sum_{j}\lambda_{j}=1$ such that%
\begin{equation}
\int_{\mathbb{R}^{2n}}H(z)\rho(z)\,\mathrm{d}z=\sum_{j}\lambda_{j}%
(\widehat{H}\psi_{j}|\psi_{j})_{L^{2}}. \label{convex}%
\end{equation}

\end{proposition}

\begin{proof}
The condition (\ref{quantum}) implies that $\rho$ i the Wigner distribution of
\ quantum state, hence there exists an orthonormal sequence of functions
$\psi_{j}\in L^{2}(\mathbb{R}^{n})$ such that \
\[
\rho=\sum_{j}\lambda_{j}W\psi_{j}\text{ \ , \ }\lambda_{j}\geq0,\sum
_{j}\lambda_{j}=1.
\]
and thus%
\[
\int_{\mathbb{R}^{2n}}H(z)\rho(z)\,\mathrm{d}z=\sum_{j}\lambda_{j}%
\int_{\mathbb{R}^{2n}}H(z)W\psi_{j}(z)\mathrm{d}z.
\]
In view of Lemma (\ref{Lemma1}) this is precisely (\ref{convex}).
\end{proof}

\section{Quadratic Hamiltonians}

Recall \cite{Birk} that the symplectic group $\operatorname*{Sp}(n)$ is
generated by $J$ together with all matrices
\[
V_{-P}=%
\begin{pmatrix}
I & 0\\
P & I
\end{pmatrix}
\text{ \ and }M_{L}=%
\begin{pmatrix}
I & 0\\
P & I
\end{pmatrix}
\]
with $P=P^{T}$ and $\det L\neq0$.

\subsection{Pure Gaussian states}

Assume that the covariance ellipsoid $\Omega$ of the Gaussian distribution
(\ref{Gauss}) is a quantum blob $Q=SB^{2n}(\sqrt{\hbar})$, $S\in
\operatorname*{Sp}(n)$, hence the covariance matrix is%
\[
\Sigma=\frac{\hbar}{2}SS^{T}\text{ }.
\]
It follows \cite{Birk,blob,goluPR,Littlejohn} that (\ref{Gauss}) is the Wigner
distribution of a Gaussian function of the type
\begin{equation}
\psi_{X,Y}(x)=\left(  \tfrac{\det X}{(\pi\hbar)^{n}}\right)  ^{1/4}%
e^{-\frac{1}{2\hbar}(X+iY)x\cdot x} \label{Gausso}%
\end{equation}
($X,Y$ real symmetric and$X>0$), generalizing the standard Gaussian standard
Gaussian
\[
\phi_{0}(x)=(\pi\hbar)^{-n/4}e^{-|x|^{2}/2\hbar}.
\]
In fact, its Wigner transform is
\begin{equation}
W\psi_{X,Y}(z)=\left(  \frac{1}{\pi\hbar}\right)  ^{n}e^{-\frac{1}{\hbar
}G_{XY}z\cdot z} \label{goupsi}%
\end{equation}
where $G$ is a positive definite symplectic matrix
\begin{equation}
G_{X,Y}=S_{X,Y}^{T}S_{X,Y}\text{ \ , \ }S_{X,Y}=%
\begin{pmatrix}
X^{1/2} & 0\\
X^{-1/2}Y & X^{-1/2}%
\end{pmatrix}
. \label{sexy}%
\end{equation}
Note that since we have
\[
G_{X,Y}=S_{X,Y}^{T}S_{X,Y}=(SS^{T})^{-1}%
\]
the initial quantum blob is just%
\[
Q=SB^{2n}(\sqrt{\hbar})=S_{X,Y}^{-1}B^{2n}(\sqrt{\hbar})
\]

It follows \cite{blob,goluPR} that there is a canonical bijective
identifications of the set of quantum blobs and the set of pure Gaussian
quantum states $|\psi_{X,Y}\rangle$:%
\begin{equation}
\Phi:S_{X,Y}^{-1}B^{2n}(B^{2n}(\sqrt{\hbar}))\longleftrightarrow\psi
_{X,Y}=\widehat{S}_{X,Y}\phi_{0} \label{bijection}%
\end{equation}
where $\phi_{0}$ is the standard Gaussian and $\widehat{S}\in
\operatorname*{Mp}(n)$ any of the two metaplectic operators covering $S$,
namely
\begin{equation}
\widehat{S}_{X,Y}=\widehat{V}_{Y}\widehat{M}_{X^{1/2}} \label{sexy2}%
\end{equation}
where $\widehat{V}_{Y}$ and $\widehat{M}_{X^{1/2}}$ are obtained from the
elementary generators of $\operatorname*{Mp}(n)$ \ given by
\cite{Birk,Birkbis}
\begin{equation}
\widehat{V}_{P}\psi(x)=e^{--\frac{i}{2\hbar}Px\cdot x}\psi(x)\text{
\ },\widehat{M}_{L}\psi(x)=\sqrt{\det L}\psi(Lx) \label{unit}%
\end{equation}
covering the symplectic matrices $\widehat{V}_{P}$ and $M_{L}$, hence we have
$\psi_{X,Y}=\widehat{S}_{XY}\phi_{0}$ with $\widehat{S}_{XY}=\widehat{V}%
_{Y}\widehat{M}_{X^{1/2}}\in\operatorname*{Mp}(n)$.

Summarizing:

\begin{proposition}
\label{prop2}Every quantum blob can be written $S_{XY}B^{2n}(\sqrt{\hbar})$
for a unique choice of matrices $X,Y$. The covariancw ellipsoid of
$\ \psi_{X,Y}$ is then precisely $S_{X,Y}^{-1}B^{2n}(\sqrt{\hbar})$.
\end{proposition}

\begin{proof}
Let $S\in\operatorname*{Sp}(n)$; using \ pre-Iwasawaa factorization
\cite{Birk,Houde} we can fin $L=L^{T}>0$ and $P=P^{T}$ such that%
\[
S=V_{-P}M_{L}U\text{ \ , \ }U\in\operatorname*{Sp}(n)\cap O(2n,\mathbb{R}).
\]
By rotational symmetry we thus have
\[
SB^{2n}(\sqrt{\hbar})=V_{-P}M_{L}B^{2n}(\sqrt{\hbar}).
\]
Solving the equation%
\[%
\begin{pmatrix}
L^{-1} & 0\\
PL^{-1} & L
\end{pmatrix}
=%
\begin{pmatrix}
X^{1/2} & 0\\
X^{-1/2}Y & X^{-1/2}%
\end{pmatrix}
\]
yields $X=L^{-2}$ and $Y=$ $X^{1/2}PX^{1/2}$. That the covariance ellipsoid of
$\ \psi_{X,Y}$ is $S_{XY}B^{2n}(\sqrt{\hbar})$ has been established above..
\end{proof}

We note that Lieb and Ostrover \cite{Lieb} prove that for every $r>0$ he
standard normalized Gaussian is the unique maximizer on the ball $B^{2n}(R)$of
the Wigner distributions of a functions $\psi\in L^{2}(\mathbb{R}^{n})$,
$||\psi||_{L^{2}}=1$.

The Gaussians $\psi_{X,Y}$ are the eigenfunctions of a quadratic partial
differential equation generalizing the harmonic oscillator. Setting
\[
\Phi(x)=-\frac{1}{2}Yx\cdot x\text{ \ , \ \ }R(x)=\exp\left(  -\frac{1}%
{2\hbar}Xx\cdot x\right)
\]
we have
\begin{equation}
\nabla_{x}\Phi(x)=-Yx\text{ \ , \ }\frac{\nabla_{x}^{2}R(x)}{R(x)}=-\frac
{1}{\hbar}\operatorname*{Tr}X+\frac{1}{\hbar^{2}}X^{2}x\cdot x \label{tr}%
\end{equation}
and one verifies that \
\[
\widehat{H}_{XY}\psi_{X,Y}=\left(  \frac{1}{2}\hbar\operatorname*{Tr}X\right)
\psi_{X,Y}%
\]
where
\[
\widehat{H}_{XY}=\frac{1}{2}(-i\hbar\nabla_{x}+Yx)^{2}+\frac{1}{2}X^{2}x\cdot
x.
\]
is the Weyl quantization of%

\begin{equation}
H_{XY}(z)=\frac{1}{2}(p+Yx)^{2}+\frac{1}{2}X^{2}x\cdot x; \label{gf3}%
\end{equation}
The latter can be rewritten as
\[
H_{XY}(z)=\frac{1}{2}M_{XY}z\cdot z
\]
where $M_{XY}=M_{YX}^{T}>0$ i given by%
\[
M_{XY}=%
\begin{pmatrix}
X^{2}+Y^{2} & Y\\
Y & I
\end{pmatrix}
\]
which can be written as the product
\begin{equation}
M_{XY}=S_{X,Y}%
\begin{pmatrix}
X & 0\\
0 & X
\end{pmatrix}
S_{X,Y}^{T}\emph{\ \ ,}\text{\ }S_{X,Y}=%
\begin{pmatrix}
X^{1/2} & 0\\
X^{-1/2}Y & X^{-1/2}%
\end{pmatrix}
. \label{mfs}%
\end{equation}
tat is

\subsection{The non-degenerate case: $M>0$}

Let $M$ be a real symmetric positive definite $2n\times2n$ matrix: $M=M^{T}%
>0$. In view of Williamson's diagonalization theorem \cite{Birk,HZ,Houde}
there exists a (non unique) $S\in\operatorname*{Sp}(n)$ such that%
\begin{align*}
M  &  =S_{M}^{T}DS_{M},\qquad D=%
\begin{pmatrix}
\Lambda & 0\\
0 & \Lambda
\end{pmatrix}
\\
\quad\Lambda &  =\operatorname{diag}(\lambda_{1},\ldots,\lambda_{n}%
),\quad\lambda_{j}>0.
\end{align*}
where the$\lambda_{j}>0$ are the symplectic eigenvalues of $JM$, that \ is the
$\pm i\lambda_{j}$ are the eigenvalues of $JM$ (that is, of the antisymmetric
matrix $M^{1/2}JM^{1/2}$).

We consider the quadratic Hamiltonian function%
\[
H(z)=\tfrac{1}{2}Mz\cdot z=\tfrac{1}{2}DS_{M}z\cdot S_{M}z.
\]
Its Weyl quantization $\widehat{H}=^{-1}\operatorname{Op}_{\mathrm{W}}(H)$ is
the operator
\[
\widehat{H}=\widehat{S_{M}}^{-1}\operatorname{Op}_{\mathrm{W}}%
(D)\widehat{S_{M}}%
\]
where $\widehat{S_{M}}\in\operatorname*{Mp}(n)$ is anyone of the two
metaplectic operators covering $S_{M}$. Explicitly%
\[
\operatorname{Op}_{\mathrm{W}}(D)=\frac{1}{2}\sum_{j=1}^{n}\lambda_{j}\left(
-\hbar^{2}\frac{\partial^{2}}{\partial x_{j}^{2}}+x_{j}^{2}\right)  .
\]

To prove our min. theorem, we use the following technical lemma:

Let $K$ be a real symmetric positive-definite $2n\times2n$ matrix. Then
\begin{equation}
\int_{\mathbb{R}^{2n}}|z|^{2}e^{-Kz\cdot z}\,dz=\frac{\pi^{n}}{2\sqrt{\det K}%
}\operatorname{Tr}(K^{-1}).
\end{equation}

\begin{lemma}
\label{lemma2}Let $K$ be a real symmetric positive-definite $2n\times2n$
matrix. Then
\begin{equation}
\int_{\mathbb{R}^{2n}}|z|^{2}e^{-Kz\cdot z}\,dz=\frac{\pi^{n}}{2\sqrt{\det K}%
}\operatorname{Tr}(K^{-1}). \label{integral}%
\end{equation}

\end{lemma}

\begin{proof}
The standard Gaussian integral is
\[
I(K)=\int_{\mathbb{R}^{2n}}e^{-Kz\cdot zz}\,dz=\pi^{n}(\det K)^{-1/2}.
\]
For $t$ near zero, $K+tI_{2n}$ is positive definite. Differentiating under the
integral sign gives
\[
\left.  \frac{d}{dt}I(K+tI_{2n})\right\vert _{t=0}=-\int_{\mathbb{R}^{2n}%
}|z|^{2}e^{-z^{T}Kz}\,dz.
\]
On the other hand, Jacobi's formula yields
\[
\frac{d}{dt}\det(K+tI_{2n})=\det(K+tI_{2n})\operatorname{Tr}((K+tI_{2n}%
)^{-1}).
\]
Hence
\[
\left.  \frac{d}{dt}I(K+tI_{2n})\right\vert _{t=0}=-\frac{\pi^{n}}{2\sqrt{\det
K}}\operatorname{Tr}(K^{-1}).
\]
Comparing the last two identities proves \eqref{integral}.
\end{proof}

We next prove a \ "symplectic trace inequality":

\begin{lemma}
\label{lemmatrace}Let $D=\operatorname{diag}(\Lambda,\Lambda)$, where
$\Lambda=\operatorname{diag}(\lambda_{1},\ldots,\lambda_{n})$ and $\lambda
_{j}>0$. For every $R\in\operatorname{Sp}(n)$,
\begin{equation}
\operatorname{Tr}(DRR^{T})\geq2\sum_{j=1}^{n}\lambda_{j}=\operatorname{Tr}D.
\label{symtrace}%
\end{equation}
Equality holds if and only if $R\in\operatorname{Sp}(n)\cap O(2n)$.
\end{lemma}

\begin{proof}
Write the rows of $R$ as $r_{1}^{T},\ldots,r_{2n}^{T}$, with each $r_{k}$ a
column vector. Since $RJR^{T}=J$, its $(j,n+j)$ entry yields
\[
r_{j}^{T}Jr_{n+j}=1,\qquad j=1,\ldots,n.
\]
As $J$ is orthogonal, the Cauchy--Schwarz and arithmetic--geometric mean
inequalities imply
\begin{equation}
1=|r_{j}^{T}Jr_{n+j}|\leq\Vert r_{j}\Vert\,\Vert r_{n+j}\Vert\leq\frac{\Vert
r_{j}\Vert^{2}+\Vert r_{n+j}\Vert^{2}}{2}. \label{eq:key}%
\end{equation}
Now $(RR^{T})_{kk}=\Vert r_{k}\Vert^{2}$, so
\[
\operatorname{Tr}(DRR^{T})=\sum_{j=1}^{n}\lambda_{j}(\Vert r_{j}\Vert
^{2}+\Vert r_{n+j}\Vert^{2})\newline\geq2\sum_{j=1}^{n}\lambda_{j}.
\]
This proves (\ref{symtrace}). For equality, since every $\lambda_{j}$ is
strictly positive, equality must hold in (\ref{symtrace}). for every $j$. Thus
$\Vert r_{j}\Vert=\Vert r_{n+j}\Vert=1$ and equality in Cauchy--Schwarz gives
$Jr_{n+j}=r_{j}$, equivalently $r_{n+j}=-Jr_{j}$. These relations say that $R$
commutes with $J$. Indeed, the $j$th row of $RJ$ is $r_{j}^{T}J$ and the
$(n+j)$th row of $R$ is $r_{n+j}^{T}=r_{j}^{T}J$; similarly the $(n+j)$th row
of $RJ$ is $-r_{j}^{T}$. Hence $RJ=JR$. Combining $RJ=JR$ with $RJR^{T}=J$
yields $JRR^{T}=J$, hence $RR^{T}=I$. Conversely, if $R$ is orthogonal, then
$RR^{T}=I$ and equality is immediate.
\end{proof}

Setting
\[
\langle H\rangle_{XY}=\int_{\mathbb{R}^{2n}}H(z)W\psi_{X,Y}(z)\,\mathrm{d}%
z=(\widehat{H}\psi_{X,Y}\mid\psi_{X,Y})_{L^{2}}%
\]
we have:

\begin{theorem}
\label{Thm1}Let $\ H$ be as quadratic Hamiltonian as above. (i) We have%
\begin{equation}
\int_{S_{X,Y}^{-1}B^{2n}(\sqrt{\hbar})}H(z)\,\,\mathrm{d}z=\frac{(\pi
\hbar)^{n}}{(n+1)!}\langle H\rangle_{XY}. \label{ave}%
\end{equation}
(ii) We have%
\begin{equation}
\langle H\rangle_{XY}=\frac{\hbar}{4}\operatorname{Tr}(MG_{XY}^{-1})
\label{trace1}%
\end{equation}
(iii) We have
\[
\min_{X>0,Y=Y^{T}}\langle H\rangle_{XY}=\frac{\hbar}{2}\sum_{j=1}^{n}%
\lambda_{j}\text{ \ , \ }\sup_{X>0,Y=Y^{T}}\langle H\rangle_{XY}=+\infty
\]

\end{theorem}

\begin{proof}
(i) The volume of a quantum blob being
\[
|SB^{2n}(\sqrt{\hbar})|=|B^{2n}(\sqrt{\hbar})|=\frac{(\pi\hbar)^{n}}{n!}%
\]
we have, by formula (\ref{av2}) in Proposition \ref{propmoy},
\[
\int_{S_{X,Y}^{-1}B^{2n}(\sqrt{\hbar})}H(z)dz==\frac{(\pi\hbar)^{n}}%
{(n+1)!}\langle H\rangle_{XY}.
\]
(ii) Let us evaluate $\langle H\rangle_{XY}$. We have
\[
\langle H\rangle_{XY}=\frac{1}{2}(\pi\hbar)^{-n}\int_{\mathbb{R}^{2n}}(Mz\cdot
z)e^{-\frac{1}{\hbar}G_{XY}z\cdot z}\mathrm{d}z
\]
and setting $u=M^{1/2}z$ and
\[
K=\ \ \frac{1}{\hbar}M^{-1/2}G_{XY}M^{-1/2}%
\]
we get
\begin{equation}
\langle H\rangle_{XY}=\frac{1}{2}(\pi\hbar)^{-n}\det M^{-1/2}\int%
_{\mathbb{R}^{2n}}(|u|^{2}e^{-Ku\cdot u}\mathrm{d}u. \label{hxy}%
\end{equation}
\ \ \ \ \ \ \ \ \ \ \ \ \ \ Applying formula (\ref{integral}) in Lemma
\ref{lemma2} yields, since $\det G_{XY}=1$,$,$%
\begin{equation}
\int_{\mathbb{R}^{2n}}|z|^{2}e^{-Kz\cdot z}\,dz=\frac{\pi^{n}}{2}\hbar
^{n}\hbar\det M^{1/2}\operatorname{Tr}(MG_{XY}^{-1}).
\end{equation}
Inserting that expression in (\ref{hxy}) we get (\ref{trace1}). (iii)
Recalling the Williamson diagonalization $M=S_{M}^{T}DS_{M}$ and the
factorization $G_{X,Y}=S_{X,Y}^{T}S_{X,Y}$ \ \ we have
\begin{align*}
\operatorname{Tr}(MG_{XY}^{-1}))  &  =\operatorname*{Tr}((S_{M}^{T}%
DS_{M})(S_{X,Y}^{-1})(S_{X,Y}^{T})^{-1})\\
&  =\operatorname*{Tr}(DS_{M}(S_{X,Y}^{-1})(S_{X,Y}^{T})^{-1}S_{M}^{T})
\end{align*}
Applying formula (\ref{symtrace}) in Lemma \ref{lemmatrace} with
$R=S_{M}(S_{X,Y}^{-1})$ we have,
\begin{equation}
\operatorname{Tr}(MG_{XY}^{-1})\geq2\sum_{j=1}^{n}\lambda_{j}.
\end{equation}

\end{proof}

\subsection{The case $M\geq0$}

Let $M=M^{T}\geq0$ be a real $2n\times2n$ matrix. Consider again
$H(z)=\frac{1}{2}z^{T}Mz$ and its Weyl quantization $\widehat{H}$.

\begin{theorem}
\label{Thm2}Let $M$ be a real symmetric positive semidefinite $2n\times2n$
matrix, and set $H(z)=\frac{1}{2}Mz\cdot z$. For $\varepsilon>0$ put
$M_{\varepsilon}=M+\varepsilon I_{2n}>0$, and denote its symplectic
eigenvalues, in non-decreasing order, by
\[
0<\lambda_{1}(\varepsilon)\leq\cdots\leq\lambda_{n}(\varepsilon).
\]
Then the limits $\lambda_{j}=\lim_{\varepsilon\downarrow0}\lambda
_{j}(\varepsilon)$ exist and are nonnegative, and
\begin{equation}
\inf_{X>0,\,Y=Y^{T}}\langle H\rangle_{X,Y}=\frac{\hbar}{2}\sum_{j=1}%
^{n}\lambda_{j}. \label{eq:infimum}%
\end{equation}

\end{theorem}

\begin{proof}
Set $H_{\varepsilon}(z)=\frac{1}{2}M_{\varepsilon}z\cdot z=H(z)+\frac
{\varepsilon}{2}|z|^{2}$ and
\[
m_{\varepsilon}=\inf_{X>0,\,Y=Y^{T}}\langle H_{\varepsilon}\rangle
_{X,Y},\qquad m=\inf_{X>0,\,Y=Y^{T}}\langle H\rangle_{X,Y}.
\]
Since $M\geq0$, we have $m\geq0$. By the positive definiteness result in
Theorem \ref{Thm1} we have,
\begin{equation}
m_{\varepsilon}=\frac{\hbar}{2}\sum_{j=1}^{n}\lambda_{j}(\varepsilon).
\label{eq:meps}%
\end{equation}

\end{proof}

For every admissible $(X,Y)$, formula (\ref{trace1}) yields
\begin{align}
\langle H_{\varepsilon}\rangle_{X,Y}  &  =\frac{\hbar}{4}\operatorname{Tr}%
((M+\varepsilon I_{2n})G_{X,Y}^{-1})\nonumber\\
&  =\langle H\rangle_{X,Y}+\frac{\hbar\varepsilon}{4}\operatorname{Tr}%
(G_{X,Y}^{-1})\geq\langle H\rangle_{X,Y}. \label{eq:comparison}%
\end{align}
Taking infima gives $m_{\varepsilon}\geq m$, hence
\begin{equation}
\liminf_{\varepsilon\downarrow0}m_{\varepsilon}\geq m. \label{eq:liminf}%
\end{equation}

Fix $\delta>0$. By the definition of the infimum, there exists an admissible
pair $(X_{\delta},Y_{\delta})$ satisfying
\begin{equation}
\langle H\rangle_{X_{\delta},Y_{\delta}}<m+\delta. \label{eq:approx}%
\end{equation}
Crucially, this pair is fixed as $\varepsilon\downarrow0$. By
\eqref{eq:comparison} and \eqref{eq:approx},
\begin{align*}
m_{\varepsilon}  &  \leq\langle H_{\varepsilon}\rangle_{X_{\delta},Y_{\delta}%
}\\
&  =\langle H\rangle_{X_{\delta},Y_{\delta}}+\frac{\hbar\varepsilon}%
{4}\operatorname{Tr}(G_{X_{\delta},Y_{\delta}}^{-1})\\
&  <m+\delta+\varepsilon C_{\delta},
\end{align*}
where $C_{\delta}=\frac{\hbar}{4}\operatorname{Tr}(G_{X_{\delta},Y_{\delta}%
}^{-1})<\infty$ is independent of $\varepsilon$. Therefore
\[
\limsup_{\varepsilon\downarrow0}m_{\varepsilon}\leq m+\delta.
\]
Since $\delta>0$ is arbitrary, combining this with \eqref{eq:liminf} gives
\begin{equation}
\lim_{\varepsilon\downarrow0}m_{\varepsilon}=m. \label{eq:convergence}%
\end{equation}
No convergence of regularized minimizers is needed. The Hermitian matrix
\[
A_{\varepsilon}=iM_{\varepsilon}^{1/2}JM_{\varepsilon}^{1/2}%
\]
has eigenvalues $-\lambda_{n}(\varepsilon),\ldots,-\lambda_{1}(\varepsilon
),\lambda_{1}(\varepsilon),\ldots,\lambda_{n}(\varepsilon)$, counted with
multiplicity. Continuity of the positive semidefinite matrix square root
implies
\[
A_{\varepsilon}\longrightarrow A_{0}=iM^{1/2}JM^{1/2}\quad\text{in operator
norm}.
\]
By Weyl's eigenvalue perturbation inequality, the ordered eigenvalues of these
Hermitian matrices converge. In particular, each $\lambda_{j}(\varepsilon)$
has a nonnegative limit $\lambda_{j}$. These are the limiting symplectic
spectral values of $M$; some may vanish. Finally, \eqref{eq:meps} and
\eqref{eq:convergence} imply
\[
m=\lim_{\varepsilon\downarrow0}m_{\varepsilon}=\frac{\hbar}{2}\sum_{j=1}%
^{n}\lim_{\varepsilon\downarrow0}\lambda_{j}(\varepsilon)=\frac{\hbar}{2}%
\sum_{j=1}^{n}\lambda_{j},
\]
which proves \eqref{eq:infimum}.

\begin{example}
[Nonattainment]Take $n=1$ and $M=\operatorname{diag}(1,0)$, so $H(x,p)=\frac
{1}{2}x^{2}$. For $X>0$ and $Y=0$, let
\[
\psi_{X}(x)=\left(  \frac{X}{\pi\hbar}\right)  ^{1/4}\exp\left(  -\frac
{Xx^{2}}{2\hbar}\right)  .
\]
Then
\[
\langle H\rangle_{X}=\frac{1}{2}\int_{\mathbb{R}}x^{2}|\psi_{X}(x)|^{2}%
\,dx=\frac{\hbar}{4X}\longrightarrow0\qquad(X\rightarrow\infty).
\]
Thus the infimum is zero, but no normalized nonzero Gaussian attains it, since
its position variance is strictly positive. Meanwhile,
\[
M_{\varepsilon}=\operatorname{diag}(1+\varepsilon,\varepsilon),\qquad
\lambda_{1}(\varepsilon)=\sqrt{\varepsilon(1+\varepsilon)}\rightarrow0,
\]
in agreement with Theorem~\ref{Thm2}.
\end{example}

\begin{remark}
[Nonattainment]For one degree of freedom let $H(x,p)=\frac{1}{2}x^{2}$, so
$M=\operatorname{diag}(1,0)$ and $\lambda_{1}=0$. The normalized Gaussian
\[
\psi_{a}(x)=\left(  \frac{a}{\pi\hbar}\right)  ^{1/4}\exp\left(  -\frac
{ax^{2}}{2\hbar}\right)  ,\quad a>0,
\]
satisfies $\langle\widehat{H}\psi_{a},\psi_{a}\rangle=\hbar/(4a)\rightarrow0$.
No normalized $L^{2}$ state attains zero energy, because $\int x^{2}%
|\psi(x)|^{2}dx=0$ implies $\psi=0$ almost everywhere.
\end{remark}

\begin{remark}
[Attainment can occur for singular $M$]For two degrees of freedom take
$H=\frac{1}{2}(x_{1}^{2}+p_{1}^{2})$, independent of $(x_{2},p_{2})$. The
Gaussian minimum is $\hbar/2$, attained by the harmonic oscillator ground
state in the first variable and any centered normalized pure Gaussian in the
second. More generally, if a semidefinite $M$ admits a symplectic Williamson
diagonalization $S_{0}^{T}MS_{0}=\operatorname{diag}(\lambda_{1}%
,\ldots,\lambda_{n},\lambda_{1},\ldots,\lambda_{n})$ with $\lambda_{j}\geq0$,
the infimum is attained. A general semidefinite matrix need not admit such a diagonalization.
\end{remark}

\subsection{No-nattainment: an example}

We return to Qin's \cite{Qin2} water-bag example with Hamiltonian function
\[
H=\frac{1}{2}(p_{x}^{2}+p_{y}^{2}+x^{2}),
\]
studied in Section \ref{sec1}. We take as initial state the normalized
standard Gaussian
\[
\psi_{0}(x,y)=\frac{1}{\sqrt{\pi\hbar}}\exp\left(  -\frac{x^{2}+y^{2}}{2\hbar
}\right)
\]
whose Wigner transform is
\[
W\psi_{0}(x,y;p_{x},p_{y})=\left(  \frac{1}{\pi\hbar}\right)  ^{2}\exp\left(
-\frac{x^{2}+y^{2}+p_{x}^{2}+p_{y}^{2}}{\hbar}\right)  .
\]
Let $\widehat{H}$ be the Weyl operator corresponding to $H$:
\[
\widehat{H}=\frac{1}{2}\left(  -\hbar^{2}\frac{\partial^{2}}{\partial x^{2}%
}-\hbar^{2}\frac{\partial^{2}}{\partial y^{2}}+x^{2}\right)  .
\]
We have
\[
(\widehat{H}\psi_{0}\mid\psi_{0})_{L^{2}}=\int_{\mathbb{R}^{4}}H(z)W\psi
_{0}(z)\mathrm{d}z
\]
where $z=(x,y,p_{x},p_{y})$. After a few calculations we this have
\[
\int_{\mathbb{R}^{4}}H(z)W\psi_{0}(z)\mathrm{d}z=\frac{1}{2}\left(
\frac{\hbar}{2}+\frac{\hbar}{2}+\frac{\hbar}{2}\right)  =\frac{3\hbar}{4}.
\]
Let $\lambda>0$ be a squeezing parameter and consider $S_{\lambda}%
\in\operatorname{Sp}(2)$ defined by
\[
S_{\lambda}(x,y,p_{x},p_{y})=(x,\lambda y,p_{x},\lambda^{-1}p_{y}).
\]
Let $\widehat{S}_{\lambda}\in\operatorname{Mp}(2)$ be anyone of the two
metaplectic operators associated with The energy of the transformed state is
therefore
\[
E(\lambda)=(\widehat{H}\widehat{S}_{\lambda}\psi_{0}|\widehat{S}_{\lambda}%
\psi_{0})_{L^{2}}=\frac{1}{2}\left(  \frac{\hbar}{2}+\frac{\hbar}{2\lambda
^{2}}+\frac{\hbar}{2}\right)  =\frac{\hbar}{2}+\frac{\hbar}{4\lambda^{2}}.
\]
Since in particular $E(1)=3\hbar/4$, we obtain, for every $\lambda>1$,
\[
E(\lambda)<E(1),\qquad\inf_{\lambda>1}E(\lambda)=\frac{\hbar}{2}%
\]
(but the infimum is not attained for any finite $\lambda>1$). Thus the state
$\psi_{0}$ is not a minimizer of the energy on its metaplectic orbit, although
the energy possesses a finite infimum along the squeezing family.

\section{Perspectives}

We have addressed the Gromov ground-state problem in a particular, yet
fundamental, setting: quadratic Hamiltonians and Gaussian quantum states
associated with quantum blobs. Although restrictive, this framework provides a
natural meeting point between symplectic geometry, phase-space methods, and
the variational structure of quantum mechanics. It also suggests several
directions for further investigation.

A first extension concerns more general phase-space ellipsoids. Instead of
restricting ourselves to quantum blobs, we may consider covariance ellipsoids
associated with arbitrary Gaussian states, including mixed states. The natural
admissibility condition is the quantum uncertainty constraint
\begin{equation}
\Sigma+\frac{i\hbar}{2}J\geq0,
\end{equation}
which ensures that the corresponding centered Gaussian phase-space
distribution is the Wigner transform of a positive density operator. Such a
distribution, denoted by $\rho$, admits a convex decomposition into pure-state
Wigner distributions,
\begin{equation}
\rho=\sum_{j}\alpha_{j}W\psi_{j},\qquad\alpha_{j}\geq0,\qquad\sum_{j}%
\alpha_{j}=1.
\end{equation}
In particular, if such a Gaussian mixed state can be represented as a
continuous convex mixture of phase-space translates of a fixed pure Gaussian
state the problem should be easily solvable for quadratic Hamiltonians using
Theorems \ref{Thm1} and \ref{Thm2}.

A more challenging question arises when one allows arbitrary decompositions
into pure states, or when the admissible states are no longer Gaussian. The
special role played by Gaussian states in the present work stems from the fact
that their Wigner distributions and covariance matrices are explicitly
controlled by symplectic transformations. For general states, the Wigner
distribution may change sign, and the covariance matrix no longer determines
the state. Extending our variational results to this setting would therefore
require additional analytical tools, possibly involving uncertainty
inequalities, phase-space concentration estimates, and the geometry of the
metaplectic representation as developed in previous work..

A second major direction is the extension to non-quadratic Hamiltonians. The
quadratic case is distinguished by the fact that the expectation of the
Weyl-quantized Hamiltonian depends only on the first and second moments of the
state. For more general Hamiltonians, higher-order features of the Wigner
distribution become relevant, and the minimization problem can no longer be
reduced to a finite-dimensional optimization over covariance matrices. It
would be particularly interesting to determine whether symplectic capacities
and related rigidity phenomena can still provide useful constraints on the
attainable energy.

Finally, an important question is whether these quantum phase-space methods
can shed new light on the original classical Gromov ground-state problem in
plasma physics. In particular, the possibility of approximating more general
phase-space distributions by suitable mixtures of Gaussian states suggests a
potential bridge between the variational theory developed here and the
energy-rearrangement problems arising in Hamiltonian plasma dynamics.
Establishing such a connection would require a careful analysis of the
admissible transformations and of the relation between quantum uncertainty
constraints and classical phase-space invariants.

These questions indicate that the Gaussian quadratic setting studied in this
paper may serve as a starting point for a broader investigation of energy
minimization under symplectic constraints, with possible applications to both
quantum mechanics and plasma physics.

\begin{acknowledgement}
We are happy to thank Prof. Leonid Polterovich (Tel Aviv) and Prof. Hong Qin
(Masison-wisconsin) or useful conversations and comments.
\end{acknowledgement}

\end{document}